\documentclass[11pt]{article}
\usepackage{latexsym}
\usepackage{tabularx,booktabs,multirow,delarray,array}
\usepackage{graphicx,amssymb,amsmath,amssymb}
\usepackage{enumerate,authblk}

\def\maxop{{\sc Max}}
\def\sumop{{\sc Sum}}

\def\annmax{ANN-\maxop}
\def\annsum{ANN-\sumop}
\newcommand{\fvd}{\mbox{$F\!V\!D$}}
\newcommand{\vd}{\mbox{$V\!D$}}

\begin{document}

\baselineskip=14.0pt

\title{\vspace*{-0.55in} Aggregate-{\sc Max} Nearest Neighbor Searching in
the Plane}

\author{
Haitao Wang\thanks{Department of Computer Science,
Utah State University, Logan, UT 84322, USA.
E-mail: {\tt haitao.wang@usu.edu}.
}
}


\date{}

\maketitle

\thispagestyle{empty}

\newtheorem{lemma}{Lemma}
\newtheorem{theorem}{Theorem}
\newtheorem{corollary}{Corollary}
\newtheorem{fact}{Fact}
\newtheorem{definition}{Definition}
\newtheorem{observation}{Observation}
\newtheorem{condition}{Condition}
\newtheorem{property}{Property}
\newtheorem{claim}{Claim}
\newenvironment{proof}{\noindent {\textbf{Proof:}}\rm}{\hfill $\Box$
\rm}

\pagestyle{plain}
\pagenumbering{arabic}
\setcounter{page}{1}

\vspace*{-0.2in}
\begin{abstract}
We study the aggregate/group nearest neighbor searching for the \maxop\
operator in the plane. For a set $P$ of $n$ points and a query set
$Q$ of $m$ points, the query asks for a point of $P$
whose maximum distance to the points in $Q$ is minimized. We present
data structures for answering such queries for both $L_1$ and $L_2$
distance measures. Previously, only heuristic and approximation algorithms were given
for both versions.  For the $L_1$ version, we build a data structure
of $O(n)$ size in $O(n\log n)$ time, such that each query can be
answered in $O(m+\log n)$ time.
For the $L_2$ version, we build a data
structure in $O(n\log n)$ time and $O(n\log \log n)$
space, such that each query can be answered in
$O(m\sqrt{n}\log^{O(1)} n)$ time, and alternatively, we
build a data structure in $O(n^{2+\epsilon})$ time and space for any $\epsilon>0$,
such that each query can be answered in $O(m\log n)$ time.
Further, we extend our result for the $L_1$ version to the top-$k$ queries where each query asks for the $k$ points of $P$ whose maximum distances to $Q$ are the smallest for any $k$ with $1\leq k\leq n$: We build a data structure
of $O(n)$ size in $O(n\log n)$ time, such that each top-$k$ query can be
answered in $O(m+k\log n)$ time.
\end{abstract}


\section{Introduction}
\label{sec:intro}

Aggregate nearest neighbor (ANN) searching
\cite{ref:AgarwalNe12,ref:LiGr11,ref:LiTw05,ref:LiFl11,ref:LianPr08,ref:LuoEf07,ref:PapadiasGr04,ref:PapadiasAg05,ref:SharifzadehVo10,ref:WangTh13arXiv,ref:YiuAg05},
also called group nearest neighbor searching,
is a generalization of the
fundamental nearest neighbor searching problem \cite{ref:AryaAn98},
where the input of each query is a set of
points and the result of the query is based on applying some
{\em aggregate} operator (e.g., \maxop\ and \sumop) on all query points.
In this paper, we consider the ANN searching on the \maxop\ operator
for both $L_1$ and $L_2$ metrics in the plane.

For any two points $p$
and $q$, let $d(p,q)$ denote the distance between $p$ and $q$.
Let $P$ be a set of $n$ points in the plane.
Given any query set
$Q$ of $m$ points, the ANN query asks for a point $p$ in $P$ such that
$g(p,Q)$ is minimized, where $g(p,Q)$ is the {\em aggregate function} of
the distances from $p$ to the points of $Q$.
The aggregate functions commonly considered are \maxop, i.e.,
$g(p,Q)=\max_{q\in Q}d(p,q)$, and \sumop,
i.e., $g(p,Q)=\sum_{q\in Q}d(p,q)$.
If the operator for $g$ is \maxop\ (resp., \sumop), we use \annmax\
(resp., \annsum) to denote the problem.

In this paper, we focus on \annmax\ in the plane for both $L_1$ and
$L_2$ versions where the distance $d(p,q)$ is measured by $L_1$ and
$L_2$ metrics, respectively.


Previously, only heuristic and approximation algorithms were given
for both versions.  For the $L_1$ version, we build a data structure
of $O(n)$ size in $O(n\log n)$ time, such that each query can be
answered in $O(m+\log n)$ time.
For the $L_2$ version, we build a data
structure in $O(n\log n)$ time and $O(n\log \log n)$
space, such that each query can be answered in
$O(m\sqrt{n}\log^{O(1)} n)$ time, and alternatively, we
build a data structure in $O(n^{2+\epsilon})$ time and space for any $\epsilon>0$,
such that each query can be answered in $O(m\log n)$ time.

Furthermore, we extend our result for the $L_1$ version to the following \annmax\ {\em top-$k$} queries. In addition to a query set $Q$, each top-$k$ query is also given an integer $k$ with $1\leq k\leq n$, and the query asks for the $k$ points $p$ of $P$ whose values $g(p,Q)$ are the smallest. We build a data structure
of $O(n)$ size in $O(n\log n)$ time, such that each $L_1$ \annmax\ top-$k$ query can be
answered in $O(m+k\log n)$ time.

\subsection{Previous Work}

For \annmax, Papadias et al. \cite{ref:PapadiasAg05}
presented a heuristic Minimum Bounding Method with worst case query time
$O(n+m)$ for the $L_2$ version. Recently, Li et al. \cite{ref:LiGr11}
gave more results on the $L_2$ \annmax\ (the queries were called {\em group
enclosing queries}). By using $R$-tree
\cite{ref:GuttmanR84}, Li et al. \cite{ref:LiGr11}
 gave an exact algorithm to answer \annmax\
queries, and the algorithm is very fast in practice but theoretically the
worst case query time is still $O(n+m)$. Li et al. \cite{ref:LiGr11}
also gave a $\sqrt{2}$-approximation algorithm with query time
$O(m+\log n)$ and the algorithm works for any fixed dimensions, and
they further extended the algorithm to obtain a
$(1+\epsilon)$-approximation result.
To the best of our knowledge, we are
not aware of any previous work that is particularly for the $L_1$
\annmax. However, Li et al. \cite{ref:LiFl11} proposed the {\em
flexible} ANN queries, which extend the classical ANN queries, and
they provided an $(1+2\sqrt{2})$-approximation algorithm that works
for any metric space in any fixed dimension.

For \annsum, a $3$-approximation solution is given in
\cite{ref:LiFl11} for the $L_2$ version. Agarwal et al.
\cite{ref:AgarwalNe12} studied nearest neighbor searching under
uncertainty, and their results can give an
$(1+\epsilon)$-approximation solution for the $L_2$ \annsum\ queries. They \cite{ref:AgarwalNe12}
also gave an exact algorithm that can solve the
$L_1$ \annsum\ problem and an improvement based on their work has been
made in \cite{ref:WangTh13arXiv}.

There are also other heuristic algorithms on ANN queries, e.g.,
\cite{ref:LiTw05,ref:LianPr08,ref:LuoEf07,ref:PapadiasGr04,ref:SharifzadehVo10,ref:YiuAg05}.

Comparing with $n$, the value $m$ is relative small in
practice. Ideally we want a solution that has a query time
$o(n)$. Our $L_1$ \annmax\ solution is the first-known exact
solution and is likely to be the best-possible. Comparing with the heuristic result
\cite{ref:LiGr11,ref:PapadiasAg05} with $O(m+n)$ worst
case query time, our $L_2$ \annmax\ solution use $o(n)$ query
time for small $m$; it should be noted that the methods in
\cite{ref:LiGr11,ref:PapadiasAg05} uses only $O(n)$ space while the
space used in our approach is larger.

In the following, we give our algorithm for the $L_1$ \annmax\ queries
in Section \ref{sec:l1} and its extension to the top-$k$ is also given
in the same section. Our result for the $L_2$ metric is presented in
Section \ref{sec:L2}. 


\section{The \annmax\ in the $L_1$ Metric}
\label{sec:l1}

In this section, we present our solution for the $L_1$ version of
\annmax\ queries as well as its extension to the top-$k$ queries. We first focus on the \annmax\ queries.
Given any query point set $Q$, our goal is to
find the point $p\in P$ such that $g(p,Q)=\max_{q\in Q}d(p,q)$ is minimized
for the $L_1$ distance $d(p,q)$, and we denote by
$\psi(Q)$ the above sought point.

For each point $p$ in the plane, denote by $p_{\max}$ the farthest
point of $Q$ to $p$.
We show below that $p_{\max}$
must be an extreme point of $Q$ along one of the four {\em diagonal} directions:
northeast, northwest, southwest, southeast.

\begin{figure}[t]
\begin{minipage}[t]{0.49\linewidth}
\begin{center}
\includegraphics[totalheight=1.5in]{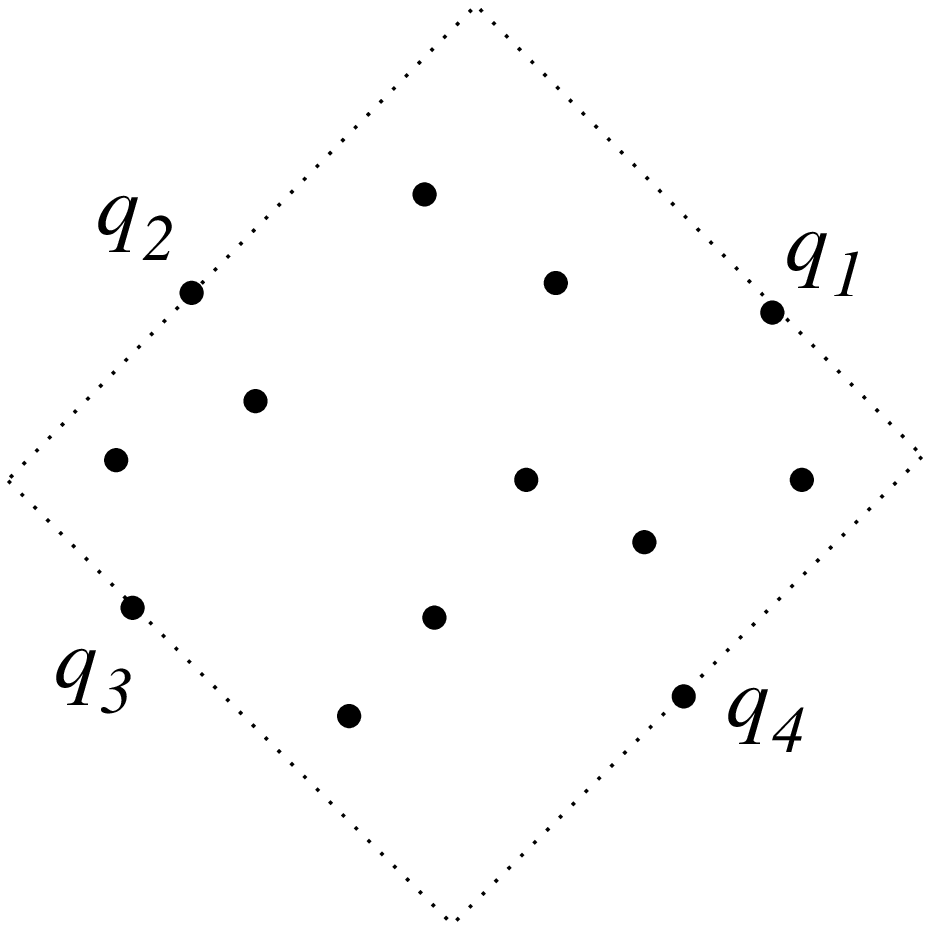}
\caption{\footnotesize Illustrating the four extreme points $q_1,q_2,q_3,q_4$.}
\label{fig:extreme}
\end{center}
\end{minipage}
\hspace*{0.05in}
\begin{minipage}[t]{0.49\linewidth}
\begin{center}
\includegraphics[totalheight=1.5in]{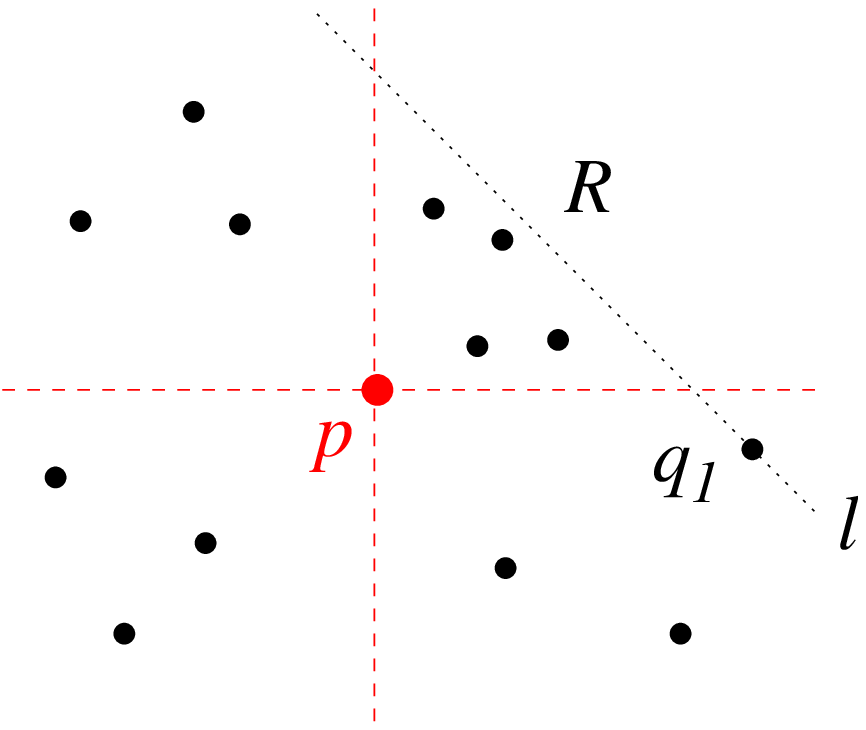}
\caption{\footnotesize Illustrating the proof for Lemma
\ref{lem:10}: $Q$ consists of all points except $p$ and $q_1$ is in
the fourth quadrant.}
\label{fig:fourth}
\end{center}
\end{minipage}
\end{figure}

Let $\rho_1$ be a ray directed to the ``northeast'', i.e., the angle
between $\rho$ and the $x$-axis is $\pi/4$. Let $q_1$ be an extreme point of
$Q$ along $\rho_1$ (e.g., see
Fig.~\ref{fig:extreme}); if there is more than one such
point, we let $q_1$ be an arbitrary such point.
Similarly, let $q_2$, $q_3$, and $q_4$ be the extreme points along
the directions northwest, southwest, and southeast, respectively. Let
$Q_{\max}=\{q_1,q_2,q_3,q_4\}$. Note that
$Q_{\max}$ may have less than four {\em distinct} points if two or more points of
$Q_{\max}$ refer to the same (physical) point of $Q$.
The following lemma shows that $g(p,Q)$ is determined only by the
points of $Q_{\max}$.

\begin{lemma}\label{lem:10}
For any point $p$ in the plane, $g(p,Q)=\max_{q\in
Q_{\max}}d(p,q)$ holds.
\end{lemma}
\begin{proof}
Let $p$ be any point in the plane. If $p_{max}\in Q_{\max}$,
the lemma simply follows, otherwise, we show below
that there exists a point $q'\in Q_{\max}$ such that
$d(p,q')\geq d(p,p_{\max})$, which proves the lemma.

The vertical line and horizontal line through the point $p$ partition
the plane into four quadrants. Without loss of generality, we assume
$p_{\max}$ is in the first quadrant (i.e., the northeast quadrant)
including its boundary, and we denote the quadrant by $R$.
Recall that $q_1\in Q_{\max}$ is an extreme point of $Q$ along the
northeast direction. Depending on whether $q_1\in R$, there are two
cases.

\begin{enumerate}
\item

If $q_1\in R$, since $q_1$ is an extreme point of $Q$ along
the northeast direction, we have $d(p,q_1)=\max_{q\in Q\cap
R}d(p,q)$. Due to $p_{max}\in Q\cap R$, $d(p,q_1)\geq d(p,p_{\max})$
holds.

\item
If $q_1\not\in R$, then since $p_{\max}$ is in $R$, by the definition
of $q_1$, $q_1$ is either in the second quadrant or in the
fourth quadrant. Without loss of generality,
we assume $q_1$ is in the fourth quadrant (e.g., see
Fig.~\ref{fig:fourth}).


Let $l$ be the line through $q_1$ with slope
$-1$ and denote by $s$ the line segment that is the intersection of
$l$ and $R$. According to the definition of the $L_1$ distance
measure, all points on $s$ have the same $L_1$ distance
to $p$, and we denote by $d(p,s)$ the $L_1$ distance between $p$ and
any point on $s$.
Since $q_1$ is an extreme point of $Q$ along the
northeast direction, all points of $Q$ are below or on the line $l$.
This implies $d(p,s)\geq d(p,q)$ for any $q\in Q\cap R$, and in
particular, $d(p,s)\geq d(p,p_{\max})$. On the other hand, $q_1$ is on
$l$ and $q_1\not\in R$, we have $d(p,q_1)\geq d(p,s)$. Hence, we
obtain $d(p,q_1)\geq d(p,p_{\max})$.
\end{enumerate}

The lemma thus follows.
\end{proof}

Based on Lemma \ref{lem:10}, for any point $p$ in the plane,
to determine $g(p,Q)$, we only need to consider the points in
$Q_{\max}$. Note that a point may have more than one farthest point in $Q$.
If $p$ has only one farthest point in $Q$,
then $p_{\max}$ is in $Q_{\max}$. Otherwise,
$p_{\max}$ may not be in $Q_{\max}$, and for
convenience we re-define $p_{\max}$ to be the farthest point of
$p$ in $Q_{\max}$.

For each $1\leq i\leq 4$, let $P_i=\{p\ |\ p_{\max}=q_i,
p\in P\}$, i.e., $P_i$ consists of the points of $P$ whose farthest points in $Q$ are $q_i$, and let $p_i$ be the nearest point of $q_i$ in
$P_i$. To find $\psi(Q)$, we have the following lemma.

\begin{lemma}\label{lem:20}
$\psi(Q)$ is the point $p_j$ for some $j$ with $1\leq j\leq 4$,
such that $d(p_j,q_j)\leq d(p_i,q_i)$ holds for any $1\leq i\leq 4$.
\end{lemma}
\begin{proof}
Recall that $\psi(Q)$ is the point $p\in P$ such that the value
$g(p,Q)=\max_{q\in
Q}d(p,q)$ is minimized. By their definitions, we have the following:
\begin{equation*}
\begin{split}
\min_{p\in P}g(p,Q)=\min_{p\in P}\max_{q\in Q}d(p,q)
=\min_{1\leq i\leq 4}\{\min_{p\in P_i}\max_{q\in Q}d(p,q)\}
=\min_{1\leq i\leq 4}\{\min_{p\in P_i}d(p,q_i)\}
=\min_{1\leq i\leq 4}\{d(p_i,q_i)\}.
\end{split}
\end{equation*}
The lemma thus follows.
\end{proof}

Based on Lemma \ref{lem:20}, to determine $\psi(Q)$, it is sufficient
to determine $p_i$ for each $1\leq i\leq 4$.
To this end, we make use of the farthest Voronoi diagram
\cite{ref:deBergCo08} of the four
points in $Q_{\max}$, which is also the
farthest Voronoi diagram of $Q$ by Lemma \ref{lem:10}.
Denote by $\fvd(Q)$ the farthest Voronoi diagram of $Q_{\max}$.
Since $Q_{\max}$ has only four points, $\fvd(Q)$ can be computed in
constant time, e.g., by an incremental approach. Each point $q\in
Q_{\max}$ defines a cell $C(q)$ in $\fvd(Q)$ such that every point $p\in C(q)$ is
farthest to $q_i$ among all points of $Q_{\max}$.
In order to compute the four points $p_i$ with $i=1,2,3,4$, we first
show in the following that each cell $C(q)$ has certain special
shapes that allow us to make use of the
segment dragging queries \cite{ref:ChazelleAn88,ref:MitchellL192}
to find the four points efficiently.
Note that for each $1\leq i\leq 4$, $P_i=P\cap C(q_i)$ and thus $p_i$
is the nearest point of $P\cap C(q_i)$ to $q_i$.
In fact, the following discussion also gives an incremental
algorithm to compute $\fvd(Q)$ in constant time.

\subsection{The Bisectors}

We first briefly discuss the bisectors of the points based on the $L_1$ metric. In fact, the $L_1$ bisectors have been well studied (e.g., \cite{ref:MitchellL192}) and we discuss them here for completeness and some notation introduced here will also be useful later when we describe our algorithm.

For any two points $q$ and $q'$ in the plane,
define $r(q,q')$ as the region of the
plane that is the locus of the points farther to $q$ than to $q'$,
i.e., $r(q,q')=\{p\ |\ d(p,q)\geq d(p,q')\}$. The {\em bisector} of
$q$ and $q'$, denoted by $B(q,q')$, is the locus of the points that
are equidistant to $q$ and $q'$, i.e., $B(q,q')=\{p\ |\ d(p,q)= d(p,q')\}$.
In order to discuss the shapes of the cells  of $\fvd(Q)$, we
need to elaborate on the shape of $B(q,q')$, as follows.

Let $R(q,q')$ be the rectangle that has
$q$ and $q'$ as its two vertices on diagonal positions (e.g., see
Fig.~\ref{fig:bisector}). In the special case where the line segment
$\overline{qq'}$ is axis-parallel, the rectangle $R(q,q')$ is
degenerated into a line segment and $B(q,q')$ is the line through the
midpoint of $\overline{qq'}$ and perpendicular to $\overline{qq'}$.
Below, we focus on the general case where $\overline{qq'}$ is
not axis-parallel. Without loss of generality, we assume $q$ and
$q'$ are northeast and southwest vertices of $R(q,q')$, and other cases are
similar.

\begin{figure}[t]
\begin{minipage}[t]{\linewidth}
\begin{center}
\includegraphics[totalheight=1.5in]{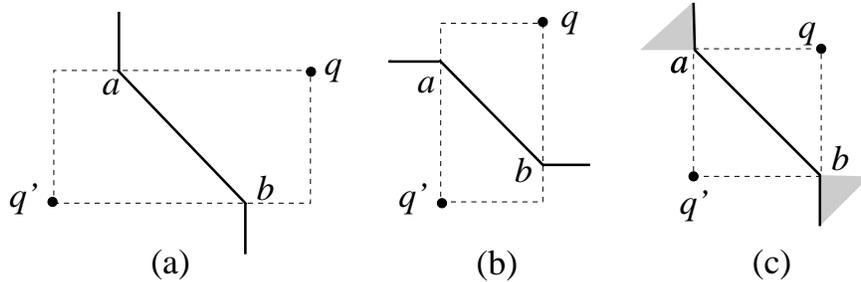}
\caption{\footnotesize Illustrating the bisector $B(q,q')$
(the solid curve) for $q$ and $q'$. In (c), since $R(q,q')$ is a
square, the two shaded quadrants are entirely in $B(q,q')$, but for simplicity, we only consider the two vertical bounding half-lines as in $B(q,q')$. }
\label{fig:bisector}
\end{center}
\end{minipage}
\end{figure}

The bisector $B(q,q')$ consists of two half-lines and one line
segment in between (e.g., see Fig.~\ref{fig:bisector}); the two
half-lines are either both horizontal or both vertical.
More specifically, let $l$ be the line of slope $-1$ that contains
the midpoint of $\overline{qq'}$. Let $\overline{ab}=l\cap R(q,q')$,
and $a$ and $b$ are on the boundary of $R(q,q')$.
Note that if $R(q,q')$ is a square, then $a$ and  $b$ are the other
two vertices of $R(q,q')$ than $q$ and $q'$; otherwise, neither $a$
nor $b$ is a vertex.

We first discuss the case where $R(q,q')$ is not a square (e.g.,
see Fig.~\ref{fig:bisector} (a) and (b)).
Let $l(a)$ be the line through $a$ and perpendicular to the edge of
$R(q,q')$ that contains $a$. The point $a$
divides $l(a)$ into two half-lines, and we let $l'(a)$ be the one
that doest not intersect $R(q,q')$ except $a$.
Similarly, we define the half-line $l'(b)$.
Note that $l'(a)$ and $l'(b)$ must be parallel.
The bisector $B(q,q')$ is the union
of $l'(a)$, $\overline{ab}$, and $l'(b)$.

If $R(q,q')$ is a square, then $a$ and $b$ are both vertices of
$R(q,q')$ (e.g., see Fig.~\ref{fig:bisector} (c)).
In this case, a quadrant of $a$ and a quadrant of $b$ belong to the
bisector $B(q,q')$,  but for simplicity, we consider $B(q,q')$ as the
union of $\overline{ab}$ and the two vertical bounding half-lines of
the two quadrants.

We call $\overline{ab}$ the {\em middle segment} of $B(q,q')$ and
denote it by $B_M(q,q')$. If $B(q,q')$ contains two vertical
half-lines, we call $B(q,q')$ a {\em v-bisector} and refer to the two
half-lines as {\em upper half-line} and
{\em lower half-line}, respectively, based on their
relative positions; similarly, if $B(q,q')$ contains two horizontal
half-lines, we call $B(q,q')$ an {\em h-bisector} and refer to the two
half-lines as {\em left half-line} and
{\em right half-line}, respectively.

For any point $p$ in the plane, we use $l^+(q)$ to denote the line
through $q$ with slope $1$, $l^-(q)$ the line
through $q$ with slope $-1$, $l_h(q)$ the horizontal line through
$q$, and $l_v(q)$ the vertical line through $q$.

\subsection{The Shapes of Cells of $\fvd(Q)$}

In the following, we discuss the shapes of the cells of $\fvd(Q)$.
A subset $Q'$ of $Q$ is {\em extreme} if it contains an extreme point along each of the four diagonal directions. The set $Q_{\max}$ is an extreme subset. A point $q$ of $Q_{\max}$ is {\em redundant} if $Q_{\max}\setminus\{q\}$ is still an extreme subset.
For simplicity of discussion, we remove all redundant points from $Q_{\max}$.
For example, if $q_1$ and $q_2$ are both extreme points along the northeast direction (and $q_2$ is also an extreme point along the northwest direction), then $q_1$ is redundant and we simply remove $q_1$ from $Q_{\max}$ (and the new $q_1$ of $Q_{\max}$ now refers to the same physical point as $q_2$).

Consider a point $q\in Q_{\max}$. Without loss of generality, we assume
$q=q_3$ and the other cases can be analyzed similarly.
We will analyze the possible shapes of $C(q_3)$.
We assume $Q_{\max}$ has at least two distinct points since otherwise the problem
would be trivial.
We further assume $q_1\neq q_3$ since otherwise the analysis is much simpler.
According to their definitions,
$q_1$ must be above the line $l^-(q_3)$ (e.g., see
Fig.~\ref{fig:position}). However, $q_1$ can be either above or below
the line $l^+(q_3)$. In the following discussion,
we assume $q_1$ is below or on the line $l^+(q_3)$ and the case
where $q_1$ is above $l^+(q_3)$ can be analyzed similarly.
In this case $B(q_3,q_1)$ is a v-bisector (i.e., it has two vertical
half-lines).

\begin{figure}[t]
\begin{minipage}[t]{\linewidth}
\begin{center}
\includegraphics[totalheight=1.5in]{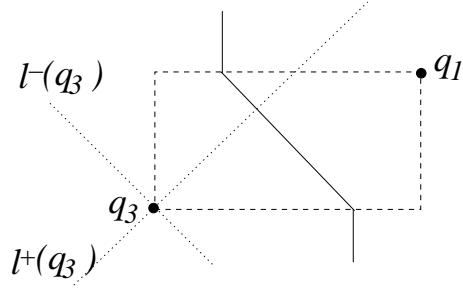}
\caption{\footnotesize Illustrating an example where $q_1$ is above
$l^-(q_3)$ and below or on $l^+(q_3)$. The bisector $B(q_1,q_3)$ is
a v-bisector (i.e., it has two vertical half-lines).}
\label{fig:position}
\end{center}
\end{minipage}
\end{figure}

We first introduce three {\em types} of regions (i.e., {\em type-A,
type-B}, and {\em type-C}), and we will show later that $C(q_3)$
must belong to one of the types.
Each type of region is bounded from
the left or below by a polygonal curve $\partial$ consisting of two
half-lines and a line segment of slope $\pm 1$ in between
(the line segment may be degenerated into a point).

\begin{figure}[t]
\begin{minipage}[t]{\linewidth}
\begin{center}
\includegraphics[totalheight=1.5in]{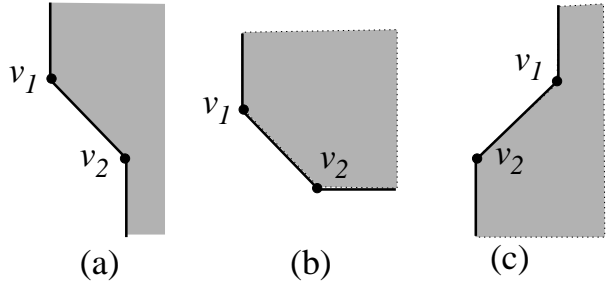}
\caption{\footnotesize Illustrating three types of regions (shaded).}
\label{fig:regiontypes}
\end{center}
\end{minipage}
\end{figure}

\begin{enumerate}
\item
From top to bottom,
the polygonal curve $\partial$ consists of a vertical half-line
followed by a line segment of slope $-1$ and then followed by a
vertical
half-line extended downwards (e.g., see Fig.~\ref{fig:regiontypes}
(a)). The
region on the right of $\partial$ is defined as a {\em type-A}
region.

\item
From top to bottom,
the polygonal curve $\partial$ consists of a vertical half-line
followed by a line segment of slope $-1$ and then followed by a
horizontal
half-line extended rightwards (e.g., see Fig.~\ref{fig:regiontypes}
(b)). The
region on the right of and above
$\partial$ is defined as a {\em type-B}
region.

\item
From top to bottom, the polygonal curve $\partial$ consists of a
vertical half-line followed by a line segment of slope $1$ and then
followed by a vertical half-line extended downwards (e.g., see
Fig.~\ref{fig:regiontypes} (c)).
The region on the right of $\partial$ is
defined as a {\em type-C} region.
\end{enumerate}

In each type of the regions,
the line segment of
$\partial$ is called the {\em middle segment}. Denote by $v_1$ the
upper endpoint of the middle segment and by $v_2$ the lower endpoint
(e.g., see  Fig.~\ref{fig:regiontypes}).  Again, the middle
segment may be degenerated to a point.
The following lemma shows that $C(q_3)$ must belong to one of the
three types of regions.

\begin{figure}[h]
\begin{minipage}[t]{\linewidth}
\begin{center}
\includegraphics[totalheight=1.5in]{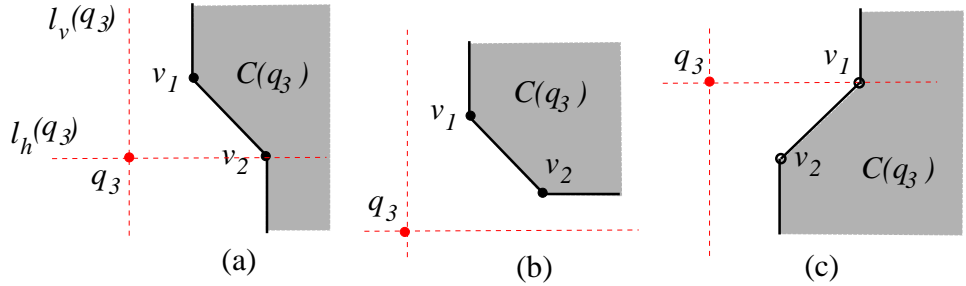}
\caption{\footnotesize Illustrating the three possible cases for
$C(q_3)$:
(a) a type-A region;
(b) a type-B region;
(c) a type-C region.}
\label{fig:cellpos}
\end{center}
\end{minipage}
\end{figure}

\begin{lemma}\label{lem:30}
The cell $C(q_3)$ must be one of the three types of regions.
Further (e.g., see Fig.~\ref{fig:cellpos}),
if $C(q_3)$ is a type-A region, then $C(q_3)$ is to the
right of $l_v(q_3)$ and $v_2$ is on $l_h(q_3)$; if $C(q_3)$ is a
type-B region, then
$C(q_3)$ is to the right of $l_v(q_3)$ and above $l_h(q_3)$; if
$C(q_3)$ a type-C region, then $C(q_3)$ is to the
right of $l_v(q_3)$ and $v_1$ is on $l_h(q_3)$.
\end{lemma}
\begin{proof}
For any point $q$ in the plane, we use $y(q)$ to denote the
$y$-coordinate of $q$ and use $x(q)$ to denote the
$x$-coordinate of $q$.

The proof is essentially an incremental approach to construct the cell
$C(q_3)$.
We first discuss the case where $y(q_1)\geq y(q_3)$ (e.g., see Fig.~\ref{fig:cell1}).
Consider the bisector $B(q_1,q_3)$, which is a v-bisector in this
case (i.e., the two half-lines of $B(q_1,q_3)$ are vertical).

First of all, if $q_1$ and $q_3$ are the only distinct points of
$Q_{\max}$, then $C(q_3)$ is $r(q_3,q_1)$ and thus $C(q_3)$ is
a type-A region. Further, $C(q_3)$ is to the right of $l_v(q_3)$ and
$v_2$ is on $l_h(q_3)$. The lemma thus follows.
Below, we assume
$q_2$ is also distinct and the case where $q_4$ is distinct is similar.

\begin{figure}[t]
\begin{minipage}[t]{0.48\linewidth}
\begin{center}
\includegraphics[totalheight=1.7in]{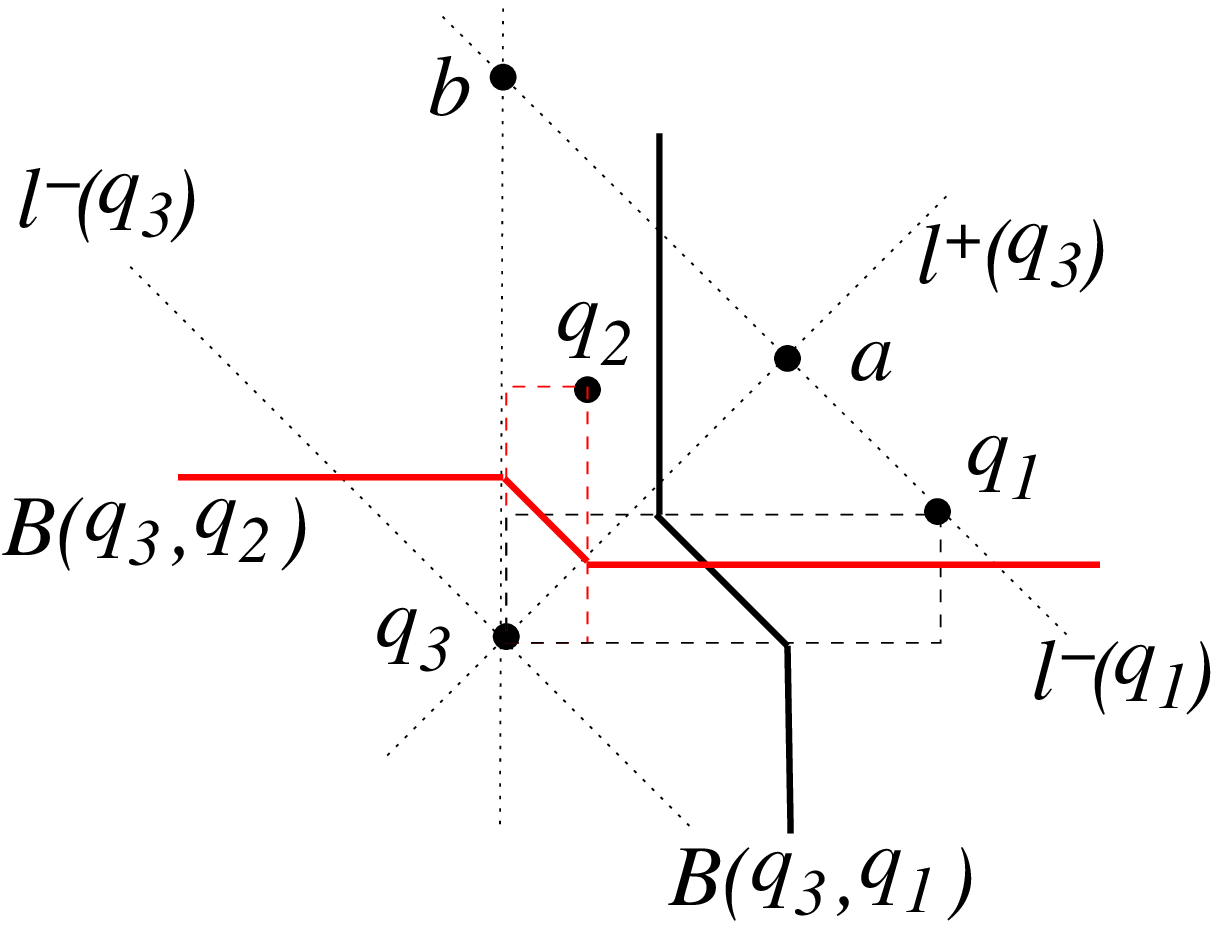}
\caption{\footnotesize Illustrating the bisector $B(q_3,q_1)$ and
the intersection $r(q_3,q_1)\cap r(q_3,q_2)$.}
\label{fig:cell1}
\end{center}
\end{minipage}
\hspace*{0.05in}
\begin{minipage}[t]{0.48\linewidth}
\begin{center}
\includegraphics[totalheight=1.7in]{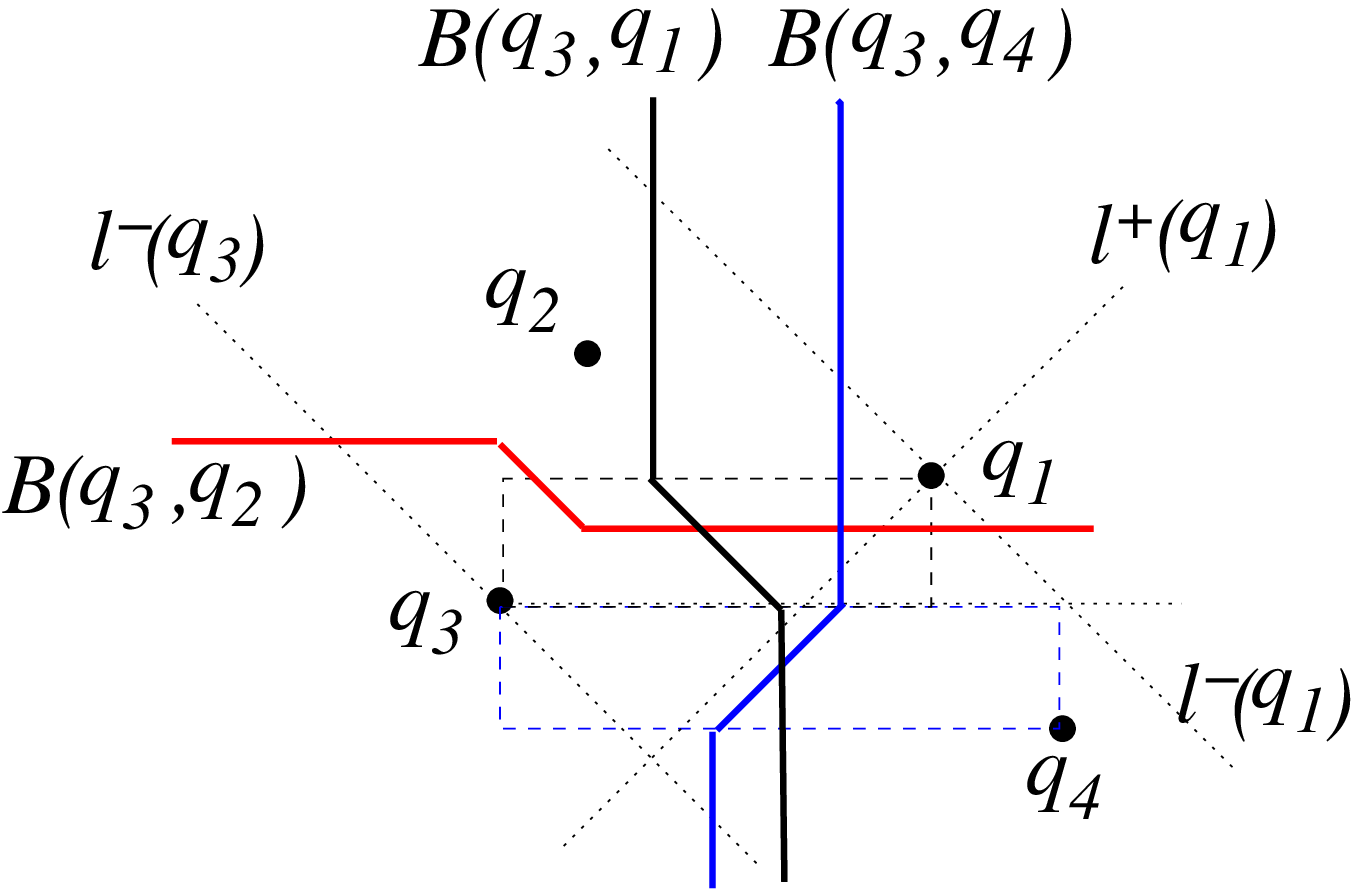}
\caption{\footnotesize Illustrating the bisector $B(q_3,q_4)$ and
the intersection $r(q_3,q_1)\cap r(q_3,q_2)\cap r(q_3,q_4)$.}
\label{fig:cell20}
\end{center}
\end{minipage}
\end{figure}

According to their definitions, $q_2$ must
be above $l^-(q_3)$, below $l^-(q_1)$, and above $l^+(q_3)$
(e.g., see Fig.~\ref{fig:cell1}); note that $q_2$ cannot be on any of
the
above three lines since otherwise $Q_{\max}$ would have a redundant point.
We analyze the shape of the intersection $r(q_3,q_1)\cap r(q_3,q_2)$.
Let $a$ be the intersection of $l^+(q_3)$ and $l^-(q_1)$.
Let $b$ be the intersection of $l_v(q_3)$ and $l^-(q_1)$.
Depending on whether $q_2$ is in the triangle $\triangle abq_3$,
there are two cases.

\begin{itemize}
\item
If $q_2\not\in \triangle abq_3$, then $x(q_2)\leq x(q_3)$. Since $q_2$ is above the line
$l^-(q_3)$, the bisector $B(q_3,q_2)$ is an h-bisector.
Since the rectangle $R(q_3,q_2)$ is on the left of $l_v(q_3)$ and
the bisector $B(q_3,q_1)$ is to the right of $l_v(q_3)$,
only the right half-line of $B(q_3,q_2)$ intersects $B(q_3,q_1)$
at a point either on the upper half-line or on the middle segment of
$B(q_3,q_1)$.
In either case, the intersection $r(q_3,q_1)\cap r(q_3,q_2)$ is a type-B
region that is above $l_h(q_3)$ and to the right of $l_v(q_3)$.

\item
If $q_3$ is strictly inside $\triangle abq_3$ (e.g., see Fig.~\ref{fig:cell1}),
then the bisector $B(q_3,q_2)$ is an h-bisector and its middle segment
$B_M(q_3,q_2)$ is of slope
$-1$. We claim that the line containing $B_M(q_3,q_2)$ is to the left of the
line containing $B_M(q_3,q_1)$. This can be proved by basic geometric
techniques, as follows.

Since $q_2$ is in the interior
of $\triangle abq_3$, we extend $\overline{q_3q_2}$ until it hits a
point on the segment $\overline{ab}$ and let $q'$ be the above point.
Since $q'$ is on $l^-(q_1)$,
the middle segment $B_M(q_3,q')$ is exactly on the
line containing $B_M(q_3,q_1)$. Since $x(q_2)< x(q')$, the claim follows.

The claim implies that the middle segment $B_M(q_3,q_2)$
does not intersect $B(q_3,q_1)$. Since the left horizontal half-line
of $B(q_3,q_2)$ is on the left of $l_v(q_3)$, it does not intersect
$B(q_3,q_1)$ either. Hence, only the right horizontal half-line of
$B(q_3,q_2)$ intersects $B(q_3,q_1)$, again at a point on the
upper half-line or the middle segment of $B(q_3,q_1)$.
Therefore, the intersection $r(q_3,q_1)\cap r(q_3,q_2)$ is a type-B
region, which is above $l_h(q_3)$ and to the right of $l_v(q_3)$.
\end{itemize}

In summary,
the intersection $r(q_3,q_1)\cap r(q_3,q_2)$ is a type-B
region that is above $l_h(q_3)$ and to the right of $l_v(q_3)$.
If there is no such a distinct point $q_4$, we are done with
proving the lemma. In the following, we assume there is a distinct
point $p_4$. Hence, the cell $C(q_3)$ is
$r(q_3,q_1)\cap r(q_3,q_2)\cap r(q_3,q_4)$.

According to their definitions, $q_4$ must be below $l^-(q_1)$, above
$l^-(q_3)$, and below $l^+(q_1)$ (e.g., see Fig.~\ref{fig:cell20}).
Note that the bisector $B(q_3,q_4)$ must be a v-bisector.
Depending on whether $y(q_4)\leq y(q_3)$, there are two cases.

\begin{itemize}
\item
If $y(q_4)\leq y(q_3)$ (e.g., see Fig.~\ref{fig:cell20}),
then each point of the rectangle $R(q_3,q_4)$ is
below or on the line $l_h(q_3)$. Hence, only the upper vertical line
of $B(q_3,q_4)$ is possible to intersect $r(q_3,q_1)\cap r(q_3,q_2)$.
Recall that $r(q_3,q_1)\cap
r(q_3,q_2)$ is a type-B region. If the vertical half-line of $B(q_3,q_4)$ intersects
$r(q_3,q_1)\cap r(q_3,q_2)$, then the cell $C(q_3)$, which is
$r(q_3,q_1)\cap r(q_3,q_2)\cap r(q_3,q_4)$, is a type-B region,
otherwise $C(q_3)=r(q_3,q_1)\cap r(q_3,q_2)$ is also a type-B region.
In either case, $C(q_3)$ is above $l_h(q_3)$ and to the right of $l_v(q_3)$.
The lemma thus follows.

\item
If $y(q_4)>y(q_3)$, then $q_4$ is in the triangle formed by
the three lines $l_h(q_3)$, $l^-(q_1)$, and $l^+(q_1)$.
The middle segment of $B(q_3,q_4)$ is of slope $-1$.
We claim that the middle segment $B_M(q_3,q_4)$
must be in the rectangle $R(q_3,q_1)$ and is to the left of the middle
segment $B_M(q_3,q_1)$. Indeed, let $d$ be the intersection of
$l_h(q_3)$ and $l^-(q_1)$. Let $e$ be the lower endpoint of
$B_M(q_3,q_1)$. By the definition of the middle segments, $e$ is the
midpoint of $\overline{q_3d}$. Since the lower edge of $R(q_3,q_4)$ is
contained in $\overline{q_3d}$ and $q_4$ is below $l^-(q_1)$,
we can obtain the claim above.

The claim implies that neither the middle segment nor the lower
vertical half-line of $B(q_3,q_4)$ can intersect $B(q_3,q_1)$.
Therefore, only the upper vertical half-line of $B(q_3,q_4)$
can intersect $r(q_3,q_1)\cap
r(q_3,q_2)$.
Therefore, as in the first case, $C(q_3)$ is a type-B regions that is
above $l_h(q_3)$ and to the right of $l_v(q_3)$.  The lemma thus follows.
\end{itemize}

We have proved the lemma for the case where $y(q_1)\geq y(q_3)$.
Next, we consider the case where $y(q_1)< y(q_3)$ (e.g., see Fig.~\ref{fig:cell40}).
The analysis is similar and we briefly discuss it below.

In this case, the middle segment $B_M(q_3,q_1)$ is of
slope $1$. If there is no other distinct point in $Q_{\max}$, then
$C(q_3)=r(q_3,q_1)$ and $C(q_3)$ is a type-C region that is to the
right of $l_v(q_3)$ and $v_1$ (i.e., the upper endpoint of the middle
segment) is on $l_h(v_3)$, which proves the lemma. Below, we assume $q_2$
is another distinct point and the case for $q_4$ is similar.

\begin{figure}[t]
\begin{minipage}[t]{0.49\linewidth}
\begin{center}
\includegraphics[totalheight=2.0in]{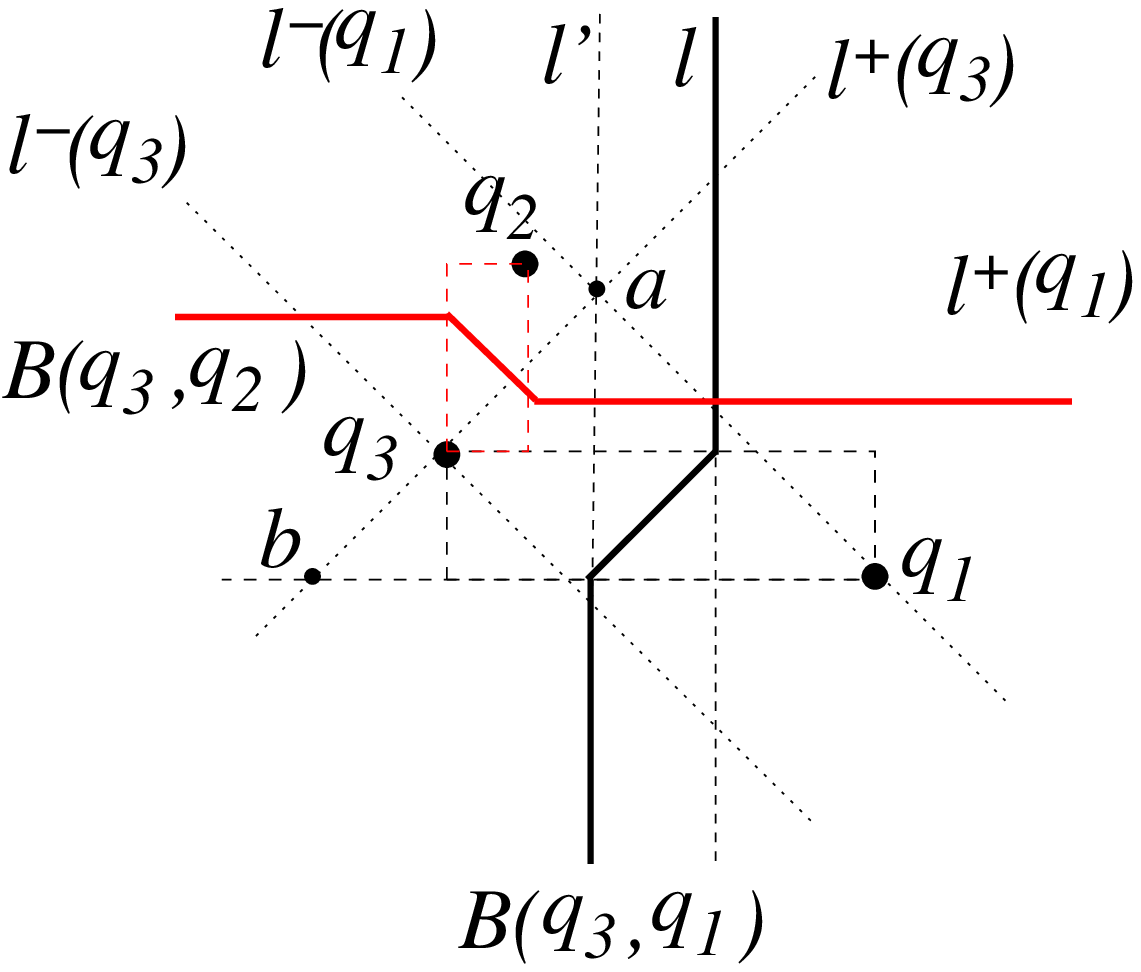}
\caption{\footnotesize Illustrating the case where $y(q_1)< y(q_3)$
and $q_2$ is shown.}
\label{fig:cell40}
\end{center}
\end{minipage}
\hspace*{-0.05in}
\begin{minipage}[t]{0.49\linewidth}
\begin{center}
\includegraphics[totalheight=2.0in]{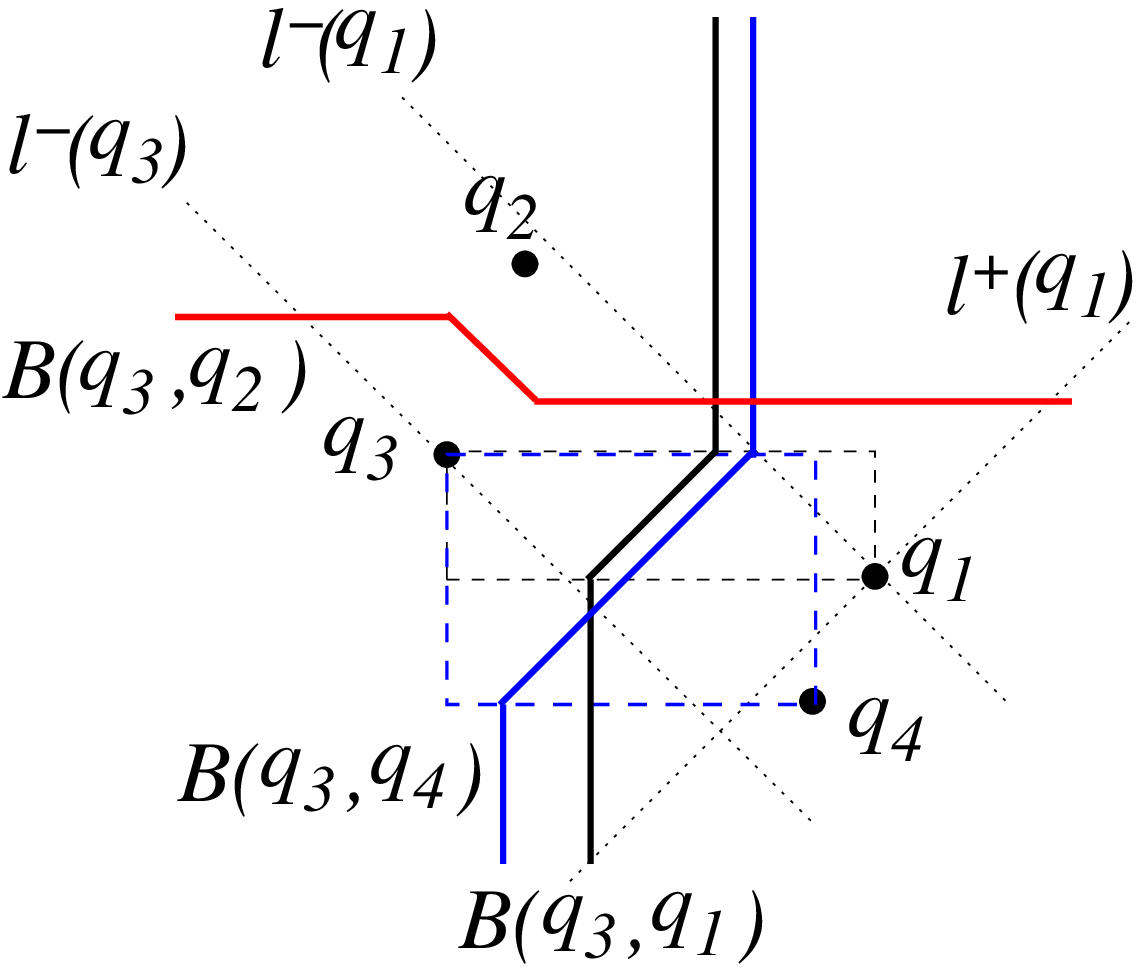}
\caption{\footnotesize Illustrating the case where $y(q_1)< y(q_3)$
and both $q_2$ and $q_4$ are shown.}
\label{fig:cell30}
\end{center}
\end{minipage}
\end{figure}

Again, $q_2$ must be above $l^-(q_3)$, below $l^-(q_1)$, and above
$l^+(q_3)$.
The bisector $B(q_3,q_2)$ is an h-bisector (e.g., see
Fig.~\ref{fig:cell40}).
Let $l$ be the vertical line
containing the upper half-line of $B(q_3,q_1)$.
We claim that $q_2$ must be to the left of $l$.
To prove the claim, it is sufficient to
show that the point $a$ is to the left of $l$,
where $a$ is the intersection of $l^-(q_1)$ and $l^+(q_3)$. To this
end, we first show that $a$ must be on $l'$, where $l'$ is the
vertical line containing the lower
vertical half-line of $B(q_3,q_1)$. To see this, consider the triangle
$\triangle abq_1$ where $b$ is the intersection of $l_h(q_1)$ and
$l^+(q_3)$. According to the definition of the middle segment of
$B(q_3,q_1)$, the lower endpoint of
$B_M(q_3,q_1)$ is the midpoint of $\overline{bq_1}$. Further,
since $\overline{ab}$ is of slope $1$ and $\overline{aq_1}$ is of
slope $-1$, the angle $\angle baq_1=\pi/2$ and the Euclidean lengths of
$\overline{ab}$ and $\overline{aq_1}$ are the same. Therefore, $a$ is
on $l'$. Since
$B_M(q_3,q_1)$ is of slope $1$, $l'$ is to the left of $l$.
The claim is proved.

Since $q_2$ must be above $l_h(q_3)$,
the above claim implies that only the right horizontal half-line of
$B(q_3,q_2)$ intersects $B(q_3,q_1)$ and the
intersection is on the upper vertical line of $B(q_3,q_1)$.
Therefore, $r(q_3,q_2)\cap r(q_3,q_1)$ is a type-B region that is
above $l_h(q_3)$ and to the right of $l_v(q_3)$. In fact,
$r(q_3,q_2)\cap r(q_3,q_1)$ is a degenerate type-B region as its
boundary consists of a vertical half-line and a horizontal half-line.
If there is no such a distinct point $q_4$ in $Q_{\max}$, we are done
with proving the lemma. In the following, we assume $q_4$ is another distinct
point, and thus $C(q_3)=r(q_3,q_2)\cap
r(q_3,q_1)\cap r(q_3,q_4)$.

Again, $q_4$ must be below $l^+(q_1)$, above
$l^-(q_3)$, and below $l^-(q_1)$ (e.g., see Fig.~\ref{fig:cell30}).
The bisector $B(q_3,q_4)$ must be a v-bisector.
Since no point of the rectangle $R(q_3,q_4)$ is above the line
$l_h(q_3)$ and $r(q_3,q_2)\cap r(q_3,q_1)$ is above $l_h(q_3)$, only
the upper vertical line of $B(q_3,q_4)$ is possible to intersect
$r(q_3,q_2)\cap r(q_3,q_1)$. Regardless of whether the upper vertical
line of $B(q_3,q_4)$ intersects $r(q_3,q_2)\cap r(q_3,q_1)$,
$C(q_3)$ is always a (degenerate) type-B region that is above
$l_h(q_1)$ and to the right of $l_v(q_3)$.

The lemma is thus proved.
\end{proof}

\subsection{Answering the Queries}

Recall that our goal is to compute $p_3$, which
is the nearest point of $P\cap C(q_3)$ to $q_3$.
Based on Lemma \ref{lem:30}, we can compute the point $p_3$ in $O(\log
n)$ time by making use of the segment dragging queries
\cite{ref:ChazelleAn88,ref:MitchellL192}. The details are given in
Lemma \ref{lem:40}.

\begin{lemma}\label{lem:40}
After $O(n\log n)$ time and $O(n)$ space preprocessing on $P$,
the point $p_3$ can be found in $O(\log n)$ time.
\end{lemma}
\begin{proof}
Before giving the algorithm, we briefly discuss the {\em segment dragging
queries} that will be used by our algorithm.

\begin{figure}[t]
\begin{minipage}[t]{\linewidth}
\begin{center}
\includegraphics[totalheight=1.2in]{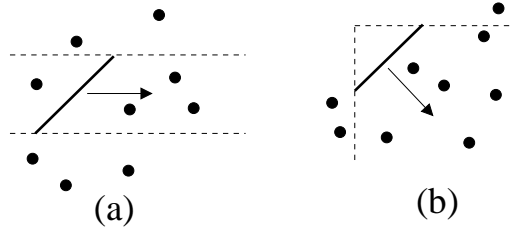}
\caption{\footnotesize Illustrating the segment dragging queries: (a)
a parallel-track query; (b) an out-of-corner query.}
\label{fig:segdrag}
\end{center}
\end{minipage}
\end{figure}

Given a set $S$ of $n$ points in the plane, we introduce two
types of segment dragging queries: the {\em parallel-track queries}
and the {\em out-of-corner queries} (e.g., Fig.~\ref{fig:segdrag}). For each
parallel-track query, we are given two parallel vertical or
horizontal lines (as ``tracks'') and a line segment
of slope $\pm 1$ with endpoints on the two tracks, and the
goal is to find the first point of $S$ hit by the segment if we drag the segment
along the two tracks. For each out-of-corner query, we are given
two axis-parallel tracks forming a perpendicular corner,
and the goal is to find the first point of
$S$ hit by dragging  out of the corner
a segment of slope $\pm 1$ with endpoints on the
two tracks.

For the parallel-track queries, as shown by Mitchell
\cite{ref:MitchellL192}, we can use Chazelle's approach
\cite{ref:ChazelleAn88} to answer each query in $O(\log n)$ time after
$O(n\log n)$ time and $O(n)$ space preprocessing on $S$. For the
our-of-corner dragging queries, by transforming it to a point
location problem, Mitchell \cite{ref:MitchellL192} gave
an algorithm that can answer each query in $O(\log n)$ time after
$O(n\log n)$ time and $O(n)$ space preprocessing on $S$.

In the sequel, we present our algorithm for the lemma by using the
above segment dragging queries. Our goal is to find $p_3$, which is
the closest point of $P\cap C(q_3)$ to $q_3$.
Depending on the type of the $C(q_3)$ as stated in Lemma \ref{lem:30},
there are three cases.

\begin{enumerate}
\item
If $C(q_3)$ is a type-A region, we further decompose $C(q_3)$ into
three subregions (e.g., see Fig.~\ref{fig:partition} (a))
by introducing two horizontal half-lines going
rightwards from $v_1$ and $v_2$ (i.e.,
the endpoints of the middle segment of the boundary of $C(q_3)$),
respectively. We call the three subregions
the {\em upper}, {\em middle}, and {\em lower} subregions,
respectively,  according to their heights. To find $p_3$, for each
subregion $C$, we compute the closest point of $P\cap C$ to $q_3$, and
$p_3$ is the closest point to $q_3$ among the three points found above.

\begin{figure}[t]
\begin{minipage}[t]{\linewidth}
\begin{center}
\includegraphics[totalheight=1.3in]{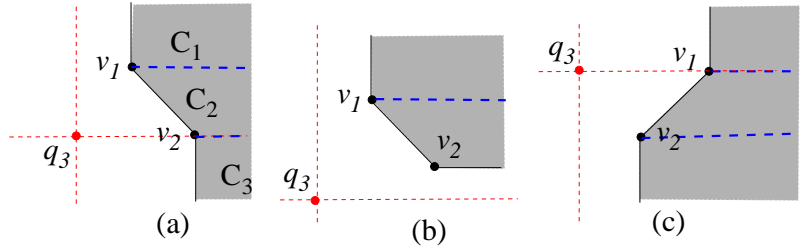}
\caption{\footnotesize Illustrating the decomposition of $C(q_3)$ for
segment-dragging queries.}
\label{fig:partition}
\end{center}
\end{minipage}
\end{figure}

For the upper subregion, denoted by $C_1$, according to Lemma
\ref{lem:30}, $C_1$ is in the first quadrant of $q_3$. Therefore,
$q_3$'s closest point in $P\cap C_1$ is exactly the answer of the
out-of-corner query by dragging a segment of slope $-1$ from the corner of $C_1$.

For the middle subregion, denoted by $C_2$, according to Lemma
\ref{lem:30}, $C_2$ is in the first quadrant of $q_3$. Therefore,
$q_3$'s closest point in $P\cap C_2$ is exactly the answer of the
parallel-track query by dragging the middle segment of the boundary of
$C(q_3)$ rightwards.

For the lower subregion, denoted by $C_3$, according to Lemma
\ref{lem:30}, $C_3$ is in the fourth quadrant of $q_3$. Therefore,
$q_3$'s closest point in $P\cap C_3$ is exactly the answer of the
out-of-corner query by dragging a segment of slope $1$ from the corner of $C_3$.

Therefore, in this case we can find $p_3$ in $O(\log n)$ time after
$O(n\log n)$ time and $O(n)$ space preprocessing on $P$.

\item
If $C(q_3)$ is a type-B region, we further decompose $C(q_3)$ into
two subregions (e.g., see Fig.~\ref{fig:partition} (b))
by introducing a horizontal half-line rightwards from $v_1$.
To find $p_3$, again, we find the closest
point to $q_3$ in each of the two sub-regions.

According to Lemma \ref{lem:30}, both subregions are in the first
quadrant of $q_3$. By using the same approach as the first case,
$q_3$'s closest point in the upper subregion can be found by an
out-of-corner query and
$q_3$'s closest point in the lower subregion can be found by a
parallel-track query.

\item
If $C(q_3)$ is a type-C region, the case is symmetric to the first
case and we can find $p_3$ by using two out-of-corner queries and a
parallel-track query.
\end{enumerate}

As a summary, we can find $p_3$ in $O(\log n)$ time after
$O(n\log n)$ time $O(n)$ space preprocessing on $P$. The lemma thus follows.
\end{proof}

By combining Lemmas \ref{lem:20} and \ref{lem:40},  we have
the following theorem.

\begin{theorem}\label{theo:10}
Given a set $P$ of $n$ points in the plane, after $O(n\log n)$ time
and
$O(n)$ space preprocessing, we can answer each $L_1$ \annmax\ query in
$O(m+\log n)$ time for any set $Q$ of $m$ query points.
\end{theorem}
\begin{proof}
As preprocessing, we build data structures for answering the segment
dragging queries on $P$ \cite{ref:ChazelleAn88,ref:MitchellL192}. The
preprocessing takes $O(n\log n)$ time and $O(n)$ space.

Given any query set $Q$, we first determine $Q_{\max}$ in $O(m)$ time.
Then, we compute the farthest Voronoi diagram $\fvd(Q)$ in constant
time, e.g., by the incremental approach given in this paper.
Then, for each $1\leq i\leq 4$, we compute the point $p_i$ by
Lemma \ref{lem:40} in $O(\log n)$ time. Finally, $\psi(Q)$ can be
determined by Lemma \ref{lem:20}.
\end{proof}

\subsection{The Top-$k$ Queries}

We extend our result in Theorem \ref{theo:10} to the top-$k$ queries.
All notations here follow those defined previously. Consider any value
$k$ with $1\leq k\leq n$.

For each point $q\in Q_{\max}$, e.g., $q=q_3$ as defined earlier, our
algorithm will find $k$ points from $P\cap C(q_3)$ nearest to $q_3$
in sorted order by their distances to $q_3$. If $|P\cap
C(q_3)|<k$, all points of $P\cap C(q_3)$ will be reported and no
other points will be reported. Due to $|Q_{\max}|\leq 4$, we will
obtain at most $4k$ points, and among them the $k$ points with the
smallest values $g(p,Q)$ are the sought points for the top-$k$ query,
which can be found in additional $O(k)$ time since the above $4k$
points are reported as four sorted lists by their values $g(p,Q)$.
The following lemma finds the $k$ points of $P\cap C(q_3)$ nearest to
$q_3$ and the algorithms for other points of $Q_{\max}$ are similar.

\begin{lemma}\label{lem:400}
After $O(n\log n)$ time and $O(n)$ space preprocessing on $P$,
the $k$ points of $P\cap C(q_3)$ nearest to $q_3$ can be found in
$O(k\log n)$ time and these points are reported in sorted order by
their distances to $q_3$.
\end{lemma}
\begin{proof}
We assume $|P\cap C(q_3)|\geq k$ since the case $|P\cap C(q_3)|< k$
can be easily solved. For ease of exposition, we also make a general
position assumption that no two points of $P$ lie on the same line of
slope $1$ or $-1$, and our approach can be generalized to
handle the general case.

We follow the discussion in the proof of Lemma \ref{lem:40}. As
preprocessing, we build the segment-dragging query data structures
\cite{ref:ChazelleAn88,ref:MitchellL192}, which takes $O(n\log n)$
time and $O(n)$ space.  Recall
that the shape of the cell $C(q_3)$ has three types. We assume
$C(q_3)$ is a type-A and the other two types can be handled
analogously.
Let $p'_1,p_2',\ldots,p_k'$ be
the $k$ points of $P\cap C(q_3)$ nearest to $q_3$
in the increasing order by their distances to $q_3$, and our algorithm
will report them in this order.

Recall that to find $p_1'$, we partition $C(q_3)$ into three
subregions $C_1$, $C_2$, and $C_3$, and for each subregion $C$, we
find the nearest point of $P\cap C$ to $q_3$; we call the above point the {\em
candidate point} for $p_1'$. Let $H$ denote the set of the above
three candidate points. The point of $H$ nearest to $q_3$ is $p_1'$.
Below we discuss how to find $p_2'$.

We first remove $p_1'$ from $H$. Next, we will find
three new candidate points and insert them to $H$, such that $p_2'$
is the nearest point of $H$ to $q_3$. The details are given below.
Depending on which subregion of $C(q_3)$ the point
$p_1'$ belongs to, there are three
cases.

\begin{enumerate}
\item If $p_1'\in C_1$, let $p$ be the second nearest point of $P\cap
C_1$ to $q_3$. It is easy to see that $p_2'$ must be one of the points
in $H\cup \{p\}$. Recall that $p_1'$ is found by dragging a segment
$s$ of slope $-1$ out of the corner of $C_1$ (i.e., $v_1$, see
Fig.~\ref{fig:partition}). After $s$
hits $p_1'$, if we keep dragging $s$, $p$ is the next point that will
be hit by $s$. To find $p$, unfortunately we cannot use the same out-of-corner
segment dragging query data structure \cite{ref:MitchellL192}
because the data structure only works
when the triangle formed by $v_1$ and $s$ does
not contain any point in its interior (see \cite{ref:MitchellL192} for
more details on this). Instead, we use the following approach.

\begin{figure}[t]
\begin{minipage}[t]{\linewidth}
\begin{center}
\includegraphics[totalheight=2.0in]{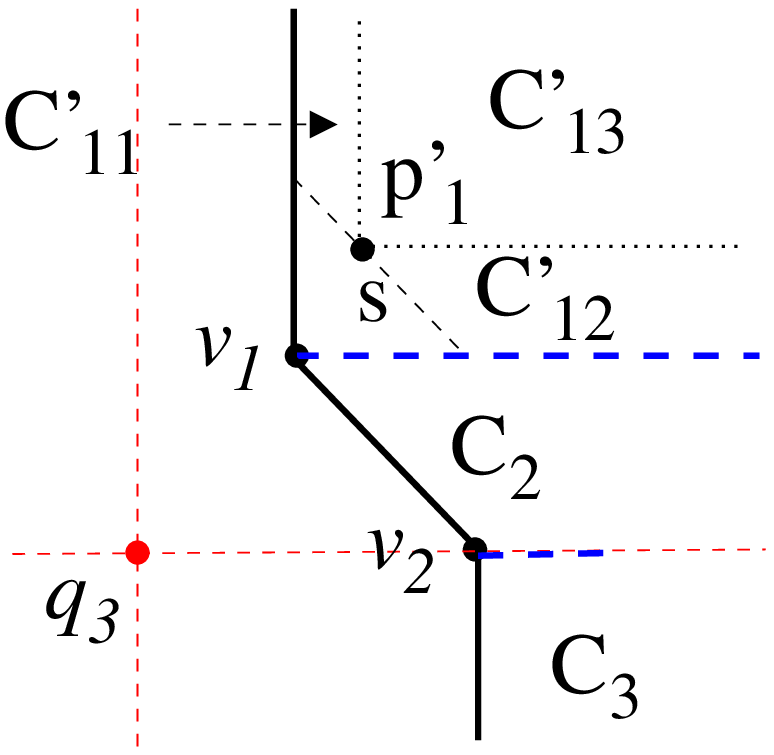}
\caption{\footnotesize Partitioning $C_1'$ into three regions:
$C'_{11}$, $C'_{12}$, and $C'_{13}$.}
\label{fig:topkcase1}
\end{center}
\end{minipage}
\end{figure}

At the moment $s$ hits $p_1'$, let $C'_1$ be the subset of $C_1$
to the right and above of $s$, i.e., $C'_1$ is $C_1$ excluding
the triangle formed by $v_1$ and $s$ (e.g., see
Fig.~\ref{fig:topkcase1}).
We partition $C'_1$ into three subregions in the following way (e.g., see
Fig.~\ref{fig:topkcase1}). Let
$C_{11}'$ be the region of $C'_1$ on the left of the vertical
line through $q_3$. Let
$C_{12}'$ be the region of $C'_1$ below the horizontal
line through $q_3$. Let $C_{13}'$ be the remaining part of $C'_1$.
For simplicity of discussion, we assume the region
$C_{i1}'$ does not contain the point $p_1'$ for any $1\leq i\leq 3$.

Denote by $p''_i$ the nearest point to $q_3$ in $P\cap C_{1i}'$, for
each $1\leq i\leq 3$. Hence,  one of $p''_i$ for $i=1,2,3$ must be
$p$, i.e., the second nearest point of $P\cup C_1$ to $q_3$.
We insert the above three points to $H$, and consequently, $p_2'$ is the point of
$H$ nearest to $q_3$. It remains to find the above three points.

The point $p_1'$ partitions $s$ into two sub-segments: let $s_1$ be
the sub-segment bounding $C'_{11}$ and $s_2$ be the one bounding
$C'_{12}$.  The point $p_1''$ can be found by a parallel-track segment dragging
query by dragging the segment $s_1$ upwards. However, there is an issue for the
approach. Since $p_1'$ is an endpoint of $s_1$, the above query may
still return $p_1'$ as the answer. We use a little trick to get around
the issue. Due to our general position assumption that no two points lie
on the same line of slope $\pm 1$, $s_1$ does not contain any other point
of $P$ than $p_1'$. Instead of dragging $s_1$, we drag another
segment $s_1'$ which can be viewed as shifting $s_1$ upwards by a
sufficiently small value $\delta$. We can determine $\delta$ in the
preprocessing step such that there is no point of $P$ strictly between
the $-1$-sloped line containing $s_1$ and the $-1$-sloped
line containing $s_1'$. For example,
one way to determine such a $\delta$ is to sort all points of $P$
by their projections to any line of slope $1$ and then find the minimum
distance between any two adjacent projections. Hence, the point
$p_1''$ is the first point hit by dragging $s_1'$ upwards.

Similarly the point $p_2''$ can also be
found by a parallel-track segment dragging query and the same trick is
applicable.

For the point $p_3''$, it can be
found by an out-of-corner segment dragging query. Note that the corner
in this case is the point $p_1'$, and thus the query may also return
$p_1'$ as the answer. This issue can also be easily resolved as
follows. The data structure in \cite{ref:MitchellL192}
for answering the out-of-corner segment dragging queries reduces the
problem into a point location problem in a planar subdivision. For the
above out-of-corner segment dragging query, we will need to locate a vertex
corresponding to $p_1'$ in the planar subdivision and the vertex is
incident to two faces: one face is for $p_1'$ and the other is for
$p_3''$. Hence, to return $p_3''$ as the answer, we only need to
report the face that does not correspond to $p_1'$.

As a summary, we can insert three new points into $H$
such that $p_2'$ is the nearest point of $H$ to $q_3$,
and the three points are found by
three segment-dragging queries, each taking $O(\log n)$ time.

\item If $p_1'\in C_2$, we use the similar approach.
Recall that $p'_1$ is found by dragging a
parallel-track segment $s$ of slope $-1$ rightwards. Let $p$ be the
second nearest point of $P\cap C_2$ to $q_3$. Clearly, $p_2'$ is the
nearest point of $H\cup \{p\}$ to $q_3$. To find $p$, after $s$ hits
$p_1'$, we can keep dragging $s$ rightwards and $p$ is the next point
that will be hit by $s$. Hence, at the moment $s$ hits $p_1'$, the
point $p$ can be found by another parallel-track segment dragging
query by dragging $s$ rightwards. Here, since $p_1'$ is on $s$,
to avoid issue that
the query returns $p_1'$ as the answer, we use the same trick as in
the first case, i.e., instead of dragging $s$, we drag a segment $s'$
that is $\delta$ distance to the right of $s$.

\item If $p_1'\in C_2$, the case is symmetric to the first case and we
omit the details.
\end{enumerate}

In summary, we can insert at most three new points into $H$
such that $p_2'$ is the nearest point of $H$ to $q_3$,
and the three points are found by
three segment-dragging queries, each taking $O(\log n)$ time.

To find the third nearest point $p_3'$, we use the similar approach. In
general, to determine $p_i'$ with $1\leq i\leq k$, we have a candidate
set $H$ such that $p_i'$ is the nearest point of $H$ to $q_3$. After
$p_i'$ is determined, we remove it from $H$, and then to find
$p_{i+1}'$, we find at most three new points by segment-dragging
queries and insert them to $H$ in the similar approach as above, such
that $p_{i+1}'$ is the nearest point of $H$ to $q_3$. We use a min-heap to
maintain the candidate set $H$, where the ``key'' of each point of $H$ is its
distance to $q_3$. Note that the size of $H$ is no more than $3k$ in the
entire algorithm and $k\leq n$.
Hence, the running time of the entire algorithm is $O(k\log n)$.
The lemma thus follows.
\end{proof}

By the preceding discussion and Lemma \ref{lem:400},
we have the following theorem.

\begin{theorem}\label{theo:20}
Given a set $P$ of $n$ points in the plane, after $O(n\log n)$ time
and
$O(n)$ space preprocessing, we can answer each $L_1$ \annmax\ top-$k$ query in
$O(m+k\log n)$ time for any set $Q$ of $m$ query points and any
integer $k$ with $1\leq k\leq n$.
\end{theorem}

\section{The \annmax\ in the $L_2$ Metric}
\label{sec:L2}

In this section, we present our results for the $L_2$ version of
\annmax\ queries. Given any query point set $Q$, our goal is to find
the point $p\in P$ such that $g(p,Q)=\max_{q\in Q}d(p,q)$ is minimized for
the $L_2$ distance $d(p,q)$, and we use $\psi(Q)$ to denote the
sought point above.

We follow the similar algorithmic scheme as in the $L_1$ version.
Let $Q_H$ be the set of points of
$Q$ that are on the convex hull of $Q$.
It is commonly known that for any point $p$ in the
plane, its farthest point in $Q$ is in $Q_H$, and in
other words, the farthest Voronoi diagram of $Q$, denoted by
$\fvd(Q)$, is determined by
the points of $Q_H$
\cite{ref:deBergCo08,ref:LiGr11}.  Note that the size of $\fvd(Q)$ is
$O(|Q_H|)$ \cite{ref:deBergCo08}.

Consider any point $q\in Q_H$. Denote by $C(q)$ the cell of $q$ in
$\fvd(Q)$. 
The cell $C(q)$ is a convex and unbounded polygon
\cite{ref:deBergCo08}. Let $f(q)$ be the closest point of $P\cap
C(q)$ to $q$. Similarly to Lemma \ref{lem:20}, we have the following
lemma.

\begin{lemma}\label{lem:42}
If for a point $q'\in Q$, $d(f(q'),q')\leq d(f(q),q)$ holds
for any $q\in Q_H$, then
$f(q')$ is $\psi(Q)$.
\end{lemma}

Hence, to find $\psi(Q)$, it is sufficient to determine $f(q)$ for
each $q\in Q$, as follows.

Consider any point $q\in Q$. To find $f(q)$, we first
triangulate the cell $C(q)$ and let $Tri(q)$ denote the
triangulation. For each triangle $\triangle\in Tri(q)$, we will find the
closest point to $q$ in $P\cap \triangle$,
denoted by $f_{\triangle}(q)$. Consequently, $f(q)$ is
the closest point to $q$ among the points $f_{\triangle}(q)$ for all
$\triangle\in Tri(q)$.

Out goal is to determine
$\psi(Q)$. To this end, we will need to triangulate each cell of
$\fvd(Q)$ and compute $f_{\triangle}(q)$ for
each $\triangle\in Tri(q)$
and for each $q\in Q$. Since the size of $\fvd(Q)$ is $O(|Q_H|)$, which is
$O(m)$, we have the following lemma.

\begin{lemma}\label{lem:45}
If the closest point $f_{\triangle}(q)$ to $q$ in $P\cap \triangle$
can be determined in $O(t_{\triangle})$ time
for any triangle $\triangle$ and any point $q$ in the plane,
then $\psi(Q)$ can be found in $O(m\cdot t_{\triangle})$ time.
\end{lemma}

In the following, we present our algorithms for computing
$f_{\triangle}(q)$  for any triangle $\triangle$ and any point $q$ in
the plane.
If we know the Voronoi diagram of the points in
$P\cap \triangle$, then $f_{\triangle}(q)$ can be determined in
logarithmic time. Hence, the problem becomes how to maintain the
Voronoi diagrams for the points in $P$
such that given any triangle $\triangle$, the Voronoi diagram
information of the
points in $P\cap\triangle$ can be obtained efficiently.
To this end, we choose to augment the $O(n)$-size simplex range
(counting) query data
structure in \cite{ref:MatousekEf92}, as shown in the following
lemma.

\begin{lemma}\label{lem:50}
After $O(n\log n\log\log n)$ time and $O(n\log\log n)$ space
preprocessing on $P$, we can compute the point $f_{\triangle}(q)$ in
$O(\sqrt{n}\log^{O(1)} n)$ time for any triangle $\triangle$
and any point $q$ in the plane.
\end{lemma}
\begin{proof}
We first briefly discuss the data structure in
\cite{ref:MatousekEf92} and then augment it for our purpose.
Note that the data structure in \cite{ref:MatousekEf92} is for any
fixed dimension and our discussion below only focuses on the planar
case, and thus each simplex below refers to a triangle.

A {\em simplicial partition} of the point set $P$ is
a collection $\Pi=\{(P_1,\triangle_1),\ldots,(P_k,\triangle_k)\}$,
where the $P_i$'s are pairwise disjoint subsets (called the {\em
classes} of $\Pi$) forming a partition of
$P$, and each $\triangle_i$ is a possibly unbounded
simplex containing the points of $P_i$. The {\em size} of $\Pi$ is $k$.
The simplex $\triangle_i$ may also contain other points in $P$
than those in $P_i$. A simplicial partition is called {\em special} if
$\max_{1\leq i\leq k}\{|P_i|\}<2\cdot \min_{1\leq i\leq k}\{|P_i|\}$,
i.e., all the classes are of roughly the same size.

The data structure in \cite{ref:MatousekEf92} is a partition tree,
denoted by $T$, based on constructing
special simplicial partitions on $P$ recursively (e.g., see
Fig.~\ref{fig:partitiontree}). The leaves of $T$ form a
partition of $P$ into constant-sized subsets.
Each internal node $v\in T$
is associated with a subset $P_v$
(and its corresponding simplex $\triangle_v$)
of $P$ and a special simplicial
partition $\Pi_v$ of size $|P_v|^{1/2}$ of $P_v$.
We assume the root of $T$ is associated with $P$ and its corresponding
simplex is the entire plane.
The {\em cardinality} of $P_v$ (i.e.,
$|P_v|$) is stored at $v$. Each internal node $v$ has
$|P_v|^{1/2}$ children that correspond to the classes of $\Pi_v$.
Thus, if $v$ is a node lying at a distance
$i$ from the root of $T$, then $|P_v|=O(n^{1/2^i})$, and the depth
of $T$ is $O(\log\log n)$. It is shown in \cite{ref:MatousekEf92} that
$T$ has $O(n)$ space and can be constructed in $O(n\log n)$ time.

\begin{figure}[t]
\begin{minipage}[t]{\linewidth}
\begin{center}
\includegraphics[totalheight=1.2in]{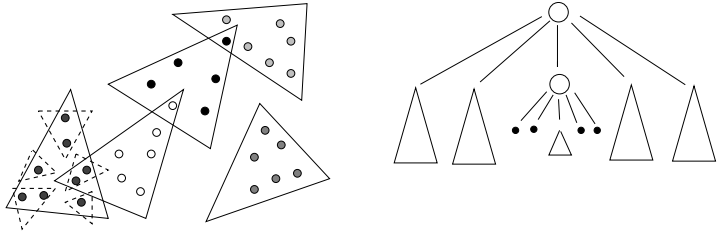}
\caption{\footnotesize Illustrating a simplicial partition and the
corresponding partition tree. The dotted triangles form the partition
computed recursively for the class corresponding to the middle child
of the root; the resulting five ``subclasses'' are stored in five subtrees below
the middle child.}
\label{fig:partitiontree}
\end{center}
\end{minipage}
\end{figure}

For each query simplex $\triangle$, the goal is to compute the number
of points in $P\cap \triangle$. We start from the root of $T$.
For each internal node $v$,
we check its simplicial partition $\Pi_v$ one by one, and handle
directly those contained in $\triangle$ or disjoint from $\triangle$; we
proceed with the corresponding child nodes for the other simplices.
Each of the latter ones must be intersected by at least one of the
lines bounding $\triangle$. If $v$ is a leaf node, for each point $p$
in $P_v$, we determine directly whether $p\in \triangle$.
Each query takes $O(n^{1/2}(\log n)^{O(1)})$
time \cite{ref:MatousekEf92}.

For our purpose, we augment the partition tree $T$ in the following
way. For each node $v$, we compute and maintain the Voronoi
diagram of $P_v$, denoted
by $\vd(P_v)$. Since at each level of $T$ the point subsets $P_v$'s
are pairwise disjoint, comparing with the original tree, our
augmented tree has $O(n)$ additional space at each level. Since $T$
has $O(\log\log n)$ levels, the total space of our augmented tree is
$O(n\log \log n)$. For the running time, we claim that the total time for building the augmented tree is still $O(n\log n)$ although we have to build Voronoi diagrams for the nodes. Indeed, let $T(n)$ denote the time for building the Voronoi diagrams in the entire algorithm. We have $T(n)=\sqrt{n}\cdot T(\sqrt{n})+O(n\log n)$, and thus, $T(n)=O(n\log n)$ by solving the above recurrence.

Consider any query triangle $\triangle$ and any point $q$.
We start from the root of $T$. For each internal node $v$,
we check its simplicial partition $\Pi_v$, i.e., check the children
of $v$ one by one.
Consider any child $u$ of $v$. If $\triangle_u$ is disjoint from
$\triangle$, we ignore it. If $\triangle_u$ is contained in
$\triangle$, then we compute in $O(\log n)$ time
the closest point of $P\cap \triangle_u$ to $q$ (and its distance to $q$) by using the Voronoi diagram $\vd(P_u)$ stored at the
node $u$. Otherwise, we proceed with the node $u$ recursively.
If $v$ is a leaf node, for each point $p$
in $P_v$, we compute directly the distance $d(q,p)$ if $p\in
\triangle$. Finally, $f_{\triangle}(q)$ is the closest point to $q$
among all points whose distances to $q$ have been computed above.

Comparing with the original simplex range query on $\triangle$, we
have $O(\log n)$ additional time on each node $u$ if
$\triangle_u$ is contained in $\triangle$, and clearly the number of
such nodes is bounded by $O(n^{1/2}(\log n)^{O(1)})$. Hence, the
total query time for finding $f_{\triangle}(q)$ is
$O(n^{1/2}(\log n)^{O(1)}\cdot \log n)$, which is
$O(n^{1/2}(\log n)^{O(1)})$.  The lemma thus follows.
\end{proof}

Similar augmentation may also be made on the $O(n)$-size
simplex data structure in
\cite{ref:MatousekRa93} and the recent randomized result in
\cite{ref:ChanOp12}. If more space are allowed, by using duality and
cutting trees \cite{ref:deBergCo08}, we can obtain the following result.

\begin{lemma}\label{lem:60}
After $O(n^{2+\epsilon})$ time and space
preprocessing on $P$, we can compute the point $f_{\triangle}(q)$ in
$O(\log n)$ time for any triangle $\triangle$ and any point
$q$ in the plane.
\end{lemma}
\begin{proof}
By using duality and cutting trees,
an $O(n^{2+\epsilon})$-size data structure can be built
in $O(n^{2+\epsilon})$ time for any $\epsilon>0$ such that each
simplex range (counting) query can be answered in $O(\log^3 n)$ time
\cite{ref:deBergCo08}. We can augment the data structure in a similar
way as in Lemma \ref{lem:50}. We only sketch it below.

The data structure in
\cite{ref:deBergCo08} has three levels. In the third
level, each tree node maintains the cardinality of the corresponding
canonical
subset of points. For our purpose, we explicitly maintain the Voronoi
diagram for each canonical subset in the third level.
Hence, our augmented data structure has four levels.
The preprocessing time and the space are the same as before.
The query algorithm is similar as before and the difference is that
when a canonical subset of points are all in the query triangle
$\triangle$, instead of counting the cardinality of the canonical
subset, we determine the closest point to $q$ in the canonical subset
by using the Voronoi diagram of the canonical subset. Hence, the total time of our query algorithm is $O(\log^4n)$ time.

To reduce the query time, a commonly known approach is to use cutting
trees with nodes having degrees $n^{\delta}$ for a certain small
constant $\delta>0$. Therefore, the heights of the trees are constant
rather than logarithmic, and consequently, the total query time
becomes $O(\log n)$. Note that we can use a point location data
structure \cite{ref:EdelsbrunnerOp86,ref:KirkpatrickOp83} to determine
in logarithmic time the child in which we continue the search, but
this does not affect the preprocessing time and space asymptotically.
The lemma thus follows.
\end{proof}

Lemmas \ref{lem:45}, \ref{lem:50}, and \ref{lem:60} together lead to
the following theorem.

\begin{theorem}
Given a set $P$ of $n$ points in the plane, after $O(n\log n\log\log n)$ time
and $O(n\log\log n)$ space preprocessing, we can answer each $L_1$ \annmax\ query in
$O(m\sqrt{n}\log^{O(1)} n)$ time for any set $Q$ of $m$ query points;
alternatively, after $O(n^{2+\epsilon})$ time and space preprocessing for
any $\epsilon>0$, we can answer each $L_2$ \annmax\ query
in $O(m\log n)$ time.
\end{theorem}

\bibliographystyle{plain}


%

\end{document}